\documentclass[11pt]{article}
\usepackage[letterpaper]{geometry}

\usepackage[cmex10]{amsmath}
\usepackage{amssymb}
\usepackage{amsfonts}
\usepackage{url}
\usepackage{amsthm}
\usepackage{fullpage}
\usepackage{framed}
\usepackage{graphicx}
\usepackage{physics}
\usepackage{mathtools}
\usepackage{color}
\usepackage{tikz}
\usetikzlibrary{intersections,calc,arrows.meta}

\usepackage[
bookmarks=true,%
bookmarksnumbered=true,%
bookmarksopen=true,%
colorlinks=true,%
linkcolor=blue,%
citecolor=blue,%
filecolor=blue,%
urlcolor=blue]{hyperref}

\newtheorem{theorem}{Theorem}
\newtheorem{corollary}{Corollary}
\newtheorem{lemma}{Lemma}

\newtheorem{definition}{Definition}

\newcommand{\ignore}[1]{}

\newcommand{\F}{\mathbb{F}}
\newcommand{\N}{\mathbb{N}}
\newcommand{\C}{\mathbb{C}}
\newcommand{\bbI}{\mathbb{I}}

\newcommand{\calH}{\mathcal{H}}
\newcommand{\calL}{\mathcal{L}}
\newcommand{\calU}{\mathcal{U}}

\newcommand{\calP}{\mathcal{P}}

\newcommand{\Exp}{\mathop{\mathbb{E}}}

\renewcommand{\epsilon}{\varepsilon}

\renewcommand{\Tr}{\mathrm{Tr}}
\renewcommand{\tr}{\mathrm{tr}}
\newcommand{\bbP}{\mathbb{P}}

\begin{document}

\title{Lower Bounds on Pauli Manipulation Detection Codes}

\author{Keiya Ichikawa\thanks{Institute of Science Tokyo \ \texttt{ichike96@gmail.com}}
\and Kenji Yasunaga\thanks{Institute of Science Tokyo \ \texttt{yasunaga@comp.isct.ac.jp}}}

\date{\today}

\maketitle

\begin{abstract}
We present a lower bound for Pauli Manipulation Detection (PMD) codes,
  a class of quantum codes that detect every Pauli error with high probability.
  Our lower bound reveals the first trade-off between the error parameter and the coding rate.
  Specifically, we show that every $q$-ary PMD code of length $n$ and coding rate $R$ 
  must satisfy $R \leq 1 - \frac{2}{n}\log_q\left(\frac{1}{\epsilon}\right) + o(1)$, 
  where  $\epsilon$ is the error parameter.
\end{abstract}

\section{Introduction}

Pauli Manipulation Detection (PMD) codes were introduced by Bergamaschi~\cite{Ber24}
as a coding scheme that detects every Pauli error with high probability.
PMD codes can be seen as a quantum analogue of Algebraic Manipulation Detection (AMD) codes~\cite{CDFPW08},
which guarantee detection of every additive error without secret keys.
Bergamaschi~\cite{Ber24}
provided an explicit construction of PMD codes based on purity testing codes~\cite{BCGST02}
and demonstrated their applications in quantum error correction and tamper detection.
Specifically, he constructed approximate quantum erasure codes approaching the quantum Singleton (or non-cloning) bound
by combining PMD codes with list-decodable stabilizer codes.
Also, he gave a construction of quantum tamper-detection codes 
for qubit-wise channels using classical non-malleable codes~\cite{DPW18}.

AMD codes have been extensively studied since their introduction in~\cite{CDFPW08}.
They provide a keyless method for detecting additive tampering
and have become a fundamental building block in information-theoretic cryptography.
Applications include robust secret sharing~\cite{CDFPW08,CPS02,CDDFS15,Che19},
robust fuzzy extractors~\cite{DORS08,DKKRS12}, multiparty computation~\cite{BT07,GIPST14,GIW16},
and non-malleable codes~\cite{DPW18}.
Tight lower bounds on the adversary's success probability and on the tag length are known~\cite{CDFPW08,CFP13,PS16},
together with near-optimal constructions~\cite{CPX15,PS16,HP18,HP19}.
In contrast, no lower bounds were previously known for PMD codes.

The ability of PMD codes to detect every Pauli error places them
within the framework of \emph{approximate} quantum error-correcting codes~\cite{LNCY97,CGS05,BGG24,YYGL24},
which can outperform \emph{exact} quantum codes by allowing a small error probability.
In particular, PMD codes provide an approximate form of quantum error detection against Pauli errors.
This detection property also relates PMD codes to recent work on quantum non-malleable codes~\cite{ABJ24,BGJR25},
and more broadly to quantum tamper-detection schemes~\cite{BK23,BB25,BKR25}.
While PMD codes focus on detecting Pauli errors,
quantum tamper detection aims to provide protection against more general classes of operations,
including broad classes of unitary attacks~\cite{BK23,BB25,BKR25}.

In this work, we present the first lower bound for PMD codes.
A $q^k$-dimensional subspace $\Pi$ of $\C^{q^n}$
is said to be an $(n, k, \epsilon)_q$-PMD code if $\| \Pi E \Pi \|_\infty \leq \epsilon$ for 
every non-identity Pauli error $E$ (see Definition~\ref{def:pmd}).
We show that every $(n, n-\lambda, \epsilon)_q$-PMD code satisfies
$\epsilon \geq \sqrt{(q^{2n - \lambda}-1)/(q^{2n} - 1)}$.
This bound reveals a trade-off between the error parameter $\epsilon$ and the coding rate $R = k/n$.
In particular, for any constant $\epsilon \in (0,1)$, every $q$-ary PMD code of length $n$, rate $R$, and error parameter $\epsilon$ must satisfy
$R \leq 1 - \frac{2}{n}\log_q\bigl(\frac{1}{\epsilon}\bigr) + o(1)$.
Compared to the parameters achieved by the $(n+\ell, n-\ell, \epsilon)_q$-PMD code of~\cite{Ber24},
which achieves $\epsilon \leq \sqrt{ (2n+1)q^{-\ell}}$,
our lower bound implies that any such code must have redundancy at least $\ell - O(\log_q n)$.
Since the redundancy of the construction in~\cite{Ber24} is $2\ell$,
this leaves a gap of $\ell + O(\log_q n)$ between our lower bound and the upper bound.
Our proof exploits the fact that the Pauli operators form 
 a unitary $1$-design, allowing us to analyze the average behavior of Pauli errors
 as if they were drawn from the entire unitary group.

\section{Preliminaries}

\subsection{Quantum States and Distances}

Let $\calL(\calH)$ be the set of linear operators on a finite Hilbert space $\calH$.
Let $A$ be a linear operator in $\calL(\calH)$.
Then, $A$ is said to be unitary if $A^\dagger A = A A^\dagger = \bbI$. 
We denote by $\calU(\calH)$ the set of all unitary operators $U \in \calL(\calH)$,
which is called the unitary group.
An operator $A$ is said to be Hermitian if $A^\dagger = A$.
A projection operator is a Hermitian operator $A$ such that $A^2 = A$.
The  trace of $A \in \calL(\calH)$ is defined as 
$\Tr(A) = \sum_{i =1}^d \bra{\vb*{e}_i} A \ket{\vb*{e}_i}$,
where $\ket{\vb*{e}_1}, \dots, \ket{\vb*{e}_d}$ form an orthonormal basis of $\calH$. 
The trace has the \emph{cyclic property} of being invariant under circular shifts; 
$\Tr(ABCD) = \Tr(BCDA) = \Tr(CDAB) = \Tr(DABC)$.
An operator $A \in \calL(\calH)$ is positive semi-definite if
$\bra{\psi} A \ket{\psi} \geq 0$ for any $\ket{\psi} \in \calH$.
A quantum state $\rho \in \calL(\calH)$ is a linear operator that is positive semi-definite and trace $1$.
We use the Schatten norms for quantifying the distances between quantum states.
The operator (or infinity) norm is $\| M \|_\infty = \max_{\ket{\psi}}\left| \bra{\psi}M^\dagger M \ket{\psi}\right|^{1/2}$,
where the maximum is taken over all unit vectors $\ket{\psi} \in \calH$.

\subsection{$q$-ary Pauli Operators}

Let $\F_q$ be a finite field of $q = p^m$ elements for a prime $p$.
The field trace is a function $\tr_{\F_q/\F_p} : \F_q \to \F_p$ such that 
$\tr_{\F_q/\F_p}(a) = \sum_{i=0}^{m-1} a^{p^i}$.
The set of elements $\{\alpha_1, \dots, \alpha_m\}$ is a basis of $\F_q$ over $\F_p$
if every $a \in \F_q$ can be expressed uniquely as $a = \sum_{i=1}^m a_i \alpha_i$, where $a_i \in \F_p$. 
A pair of bases $\alpha = \{\alpha_1, \dots, \alpha_m\}$ and $\beta = \{\beta_1, \dots, \beta_m\}$
are said to be dual bases if $\tr_{\F_q/\F_p}(\alpha_i \beta_j) = \delta_{ij}$ for every $i, j \in [m]= \{1, \dots, m\}$,
where $\delta_{ij} = 1$ if $i = j$, and $\delta_{ij}=0$ otherwise.
When $a, b \in \F_q$ are expressed as
$(a_1, \dots, a_m)$ and $(b_1, \dots, b_m)$ in the dual bases $\alpha$ and $\beta$, respectively,
the inner product becomes the field trace;
\[ \langle a, b \rangle = \sum_{i=1}^m a_ib_i 
= \sum_{i=1}^m\sum_{j=1}^m a_ib_j \tr_{\F_q/\F_p}(\alpha_i \beta_j) = \tr_{\F_q/\F_p}(ab).
\]

We define the shift operator $T$ and the phase operator $R$ over $\C^p$ as
\[ T = \sum_{x \in \F_p} \ket{x+1}\bra{x}  \text{  and  } R = \sum_{x \in \F_p} \omega^x \ket{x}\bra{x},\]
where  $\omega = e^{2\pi i/p}$.
The operators $T^iR^j$ for $i, j \in \F_p$ are said to be the Weyl-Heisenberg operators
and form an orthonormal basis of operators over $\C^p$.
If $a, b \in \F_q$ are expressed as $(a_1, \dots, a_m)$ and $(b_1, \dots, b_m)$ in the dual bases $\alpha$ and $\beta$, respectively,
we can define a basis of operators over $\C^q$ by
\[ E_{a,b} = X^aZ^b = \bigotimes_{i \in [m]}T^{a_i}R^{b_i},\]
where $\otimes$ is the tensor product.
Then, we have $E_{a,b}E_{a',b'} = \omega^{\langle a, b'\rangle - \langle a',b \rangle}E_{a',b'}E_{a,b}$.
For $\vb{a} = (a^{(1)}, \dots, a^{(n)}), \vb{b} = (b^{(1)}, \dots, b^{(n)}) \in \F_q^n$,
we can define operators on $\C^{q^n}$ by
${E}_{\vb{a},\vb{b}} = \bigotimes_{j \in [n]}E_{a^{(j)}, b^{(j)}}$.
The set of $n$ qudit Pauli operators $\bbP^n_q$ is 
$\{ {E}_{\vb{a},\vb{b}} : \vb{a}, \vb{b} \in \F_q^n \}$,
and the $n$ qudit Pauli group $\calP_q^n$ is the group generated by ${E}_{\vb{a},\vb{b}}$
and $\omega^{1/2} \cdot \bbI_{q^n \times  q^n}$.

\subsection{Haar Measure and Unitary Designs}

For a unitary group $\calU(\C^d)$ for $d \geq 1$,
the Haar measure on $\calU(\C^d)$ is the unique probability measure $\mu_H$
such that for every integrable function $f$ and every $V \in \calU(\C^d)$,
\[ \int_{\calU(\C^d)} f(U)d\mu_H(U) = \int_{\calU(\C^d)} f(UV)d\mu_H(U) = \int_{\calU(\C^d)} f(VU)d\mu_H(U).\]
Since it is a probability measure, $\int_S d\mu_H(U) \geq 0$ 
for any $S \subseteq \calU(\C^d)$ and $\int_{\calU(\C^d)} d\mu_H(U) =1$.
The expected value of $f(U)$ on $\mu_H$ is 
\[ \Exp_{U \sim \mu_H}[f(U)] = \int_{\calU(\C^d)} f(U)d\mu_H(U).\]

A probability distribution $\nu$ over $\calU(\C^d)$ is called 
a \emph{unitary $k$-design} if for every $O \in \calL(\C^{{d}^k})$, it holds that
\[ \Exp_{V \sim \nu}\left[V^{\otimes k} O {V^\dagger}^{\otimes k}\right] 
= \Exp_{U \sim \mu_H}\left[U^{\otimes k} O {U^\dagger}^{\otimes k}\right] 
.\]
If $\nu$ is the uniform distribution over a finite set $S \subseteq \calU(\C^{d})$, 
the left-hand side  is equivalent to
\[ \frac{1}{|S|} \sum_{V \in S}V^{\otimes k} O {V^\dagger}^{\otimes k}.\]
Intuitively, a unitary design is a distribution over unitaries
whose moments up to order $k$ match those of the Haar measure. 
In particular, for $k=1$, the Haar average has a simple closed form.
\begin{lemma}\cite[Corollary~13]{Mele24}\label{lem:moment}
For every $O \in \calL(\C^d)$, it holds that
\[ \Exp_{U \sim \mu_H}\left[U O {U^\dagger}\right] = \frac{\Tr(O)}{d}\bbI_{d \times d}.\]
\end{lemma}
It is well known that the uniform distribution over the Pauli operators $\bbP_q^n$ forms a unitary $1$-design,
leading to the next lemma, which will be used in our proof.
\begin{lemma}\label{lem:unitary_design}
    For every $O \in \calL(\C^{q^n})$,
\[
\frac{1}{\left| \bbP_q^n\right|} \sum_{E \in \bbP_q^n}E O E^\dagger = \frac{\Tr(O)}{q^n} \bbI_{q^n \times q^n}.
\]
\end{lemma}
It is known that every finite set 
$S \subseteq \calU(\C^d)$ whose uniform distribution forms a unitary $1$-design
satisfies $|S| \geq d^2$~\cite{RS09}.
Since $|\bbP_q^n| = q^{2n}$, the Pauli operators $\bbP_q^n$ 
give an example of a unitary $1$-design with minimum possible support size.

\section{PMD Codes and Their Lower Bounds}

A Pauli manipulation detection (PMD) code is defined as follows.
\begin{definition}\label{def:pmd}
A projection operator $\Pi$ on a $q^k$-dimensional subspace of $\C^{q^n}$ is
said to be an \emph{$(n,k,\epsilon)_q$-PMD} code if for every non-trivial
Pauli operator $E \in \calP_q^n \setminus \{ \bbI_{q^n \times q^n} \}$,
\[ \| \Pi E \Pi \|_\infty \leq \epsilon.\]
\end{definition}
We also denote by $\Pi$  the code space defined by the projection $\Pi$.
With this definition, we can see that any code state $\ket{\psi_1}$ corrupted
by a non-trivial Pauli operator $E$ is almost orthogonal to the code space.
Namely, for any code state $\ket{\psi_2} \in \Pi$,
\begin{align*}
  |\bra{\psi_2} E\ket{\psi_1}| = |\bra{\psi_2}\Pi E \Pi\ket{\psi_1}|
   \le \|\Pi E \Pi\|_\infty
   \le \epsilon.
\end{align*}

We prove a lower bound on $\epsilon$ for any PMD code. 

\begin{theorem}\label{thm:lowerbound}
    Let $\Pi$ be an $(n, n-\lambda,\epsilon)_q$-PMD code. 
    Then, it holds that
    \[ \epsilon \geq \sqrt{\frac{q^{2n - \lambda}-1}{q^{2n}-1}}.\]
\end{theorem}
\begin{proof}
We consider the following value to derive our bound:
\begin{equation}
\max_{\ket{\psi}} \Exp_{E \in \bbP_q^n}\left| \bra{\psi}\Pi E^\dagger \Pi E \Pi \ket{\psi} \right|,\label{eq:maxrandom}
\end{equation}
where the maximum is taken over all unit vectors $\ket{\psi} \in \C^{q^n}$.
First, we evaluate (\ref{eq:maxrandom}) as follows:
\begin{align}
 \max_{\ket{\psi}} \Exp_{E \in \bbP_q^n}\left| \bra{\psi}\Pi E^\dagger \Pi E \Pi \ket{\psi} \right|
& = \max_{\ket{\psi}} \Exp_{E \in \bbP_q^n} \Tr\left(\bra{\psi}\Pi E^\dagger \Pi E \Pi \ket{\psi} \right) \label{eq:1} \\ 
& = \max_{\ket{\psi}}\frac{1}{|\bbP_q^n|} \sum_{E \in \bbP_q^n}\Tr\left(\bra{\psi}\Pi E^\dagger \Pi E \Pi \ket{\psi} \right)  \nonumber\\
& = \max_{\ket{\psi}}\frac{1}{|\bbP_q^n|} \sum_{E \in \bbP_q^n}\Tr\left(\Pi E \Pi \ket{\psi} \bra{\psi}\Pi E^\dagger  \right) & \because \text{The cyclic property} \nonumber\\
& = \max_{\ket{\psi}} \Tr\left( \Pi \frac{1}{|\bbP_q^n|} \sum_{E \in \bbP_q^n}\left(E \Pi \ket{\psi} \bra{\psi}\Pi E^\dagger  \right)\right) & \because \text{The linearity} \nonumber\\
& = \max_{\ket{\psi}} \Tr\left( \Pi \frac{\Tr\left(\Pi \ket{\psi} \bra{\psi}\Pi \right)}{q^n}\right)& \because \text{Lemma~\ref{lem:unitary_design}} \nonumber\\
& = \max_{\ket{\psi}} \frac{1}{q^n} \Tr\left(\Pi \ket{\psi} \bra{\psi}\Pi \right) \Tr(\Pi)  \nonumber\\
& = q^{-\lambda}. & \because \Tr(\Pi) = q^{n-\lambda} \label{eq:2}
\end{align}
where 
(\ref{eq:1}) follows from the fact that the inner products take non-negative values
and that $a = \Tr(a)$ for $a \geq 0$.
Next, we derive  an upper bound on (\ref{eq:maxrandom}) using that $\Pi$ is  an $(n, n-\lambda, \epsilon)_q$-PMD:
\begin{align}
    \max_{\ket{\psi}} \Exp_{E \in \bbP_q^n}\left| \bra{\psi}\Pi E^\dagger \Pi E \Pi \ket{\psi} \right| 
    & = \max_{\ket{\psi}} \frac{1}{|\bbP_q^n|} \sum_{E \in \bbP_q^n}\left| \bra{\psi}\Pi E^\dagger \Pi E \Pi \ket{\psi} \right|\nonumber\\
    & \leq \frac{1}{|\bbP_q^n|} \sum_{E \in \bbP_q^n}\max_{\ket{\psi}}\left| \bra{\psi}\Pi E^\dagger \Pi E \Pi \ket{\psi} \right|\nonumber\\
    & = \frac{1}{|\bbP_q^n|} \sum_{E \in \bbP_q^n} \| \Pi E \Pi \|^2_\infty\nonumber \\
    & \leq \frac{1}{|\bbP_q^n|} \left(1 + (|\bbP_q^n| -1)\epsilon^2 \right) \label{eq:3} \\
    & = \frac{1}{q^{2n}} \left(1 + (q^{2n} -1)\epsilon^2 \right),\label{eq:4}
\end{align}
where (\ref{eq:3}) follows from the fact that $\| \Pi E \Pi \|_\infty \leq \epsilon$ 
for every $E \in \bbP_q^n \setminus \{ \bbI_{q^n \times q^n} \}$.
The statement follows from  (\ref{eq:2}) and (\ref{eq:4}).
\end{proof}
As a corollary, we obtain a lower bound on the parameter $\lambda$ using $\epsilon$, $q$, and $n$.
\begin{corollary}\label{cor:lowerbound}
    For every $(n, n-\lambda,\epsilon)_q$-PMD code, 
    it holds that
    \[ \lambda \geq 2 \log_q\left( \frac{1}{\epsilon}\right) - \frac{1-\epsilon^2}{\epsilon^2 q^{2n}\ln q}.\]
\end{corollary}
\begin{proof}
Theorem~\ref{thm:lowerbound} implies that
\begin{align*}
  q^\lambda & \geq \frac{q^{2n}}{(q^{2n}-1)\epsilon^2 + 1}
   = \frac{1}{\epsilon^2} - \frac{1/\epsilon^{2} - 1}{(q^{2n}-1)\epsilon^2 +1}
   = \frac{1}{\epsilon^2}\left( 1 - \frac{1/\epsilon^2-1}{q^{2n}+1/\epsilon^2-1}\right)\\
\end{align*}
By taking logarithms, 
\begin{align*} 
  \lambda  \geq 2\log_q\left(\frac{1}{\epsilon}\right) + \log_q\left( 1 - \frac{1/\epsilon^2-1}{q^{2n}+1/\epsilon^2-1}\right)
   \geq 2\log_q\left(\frac{1}{\epsilon}\right) - \frac{1-\epsilon^2}{\epsilon^2 q^{2n}\ln q},
\end{align*}
where the last inequality follows from the inequality $\log_q(1-x) \geq -\frac{x}{(1-x)\ln q}$ for $0 < x < 1$.
\end{proof}
Corollary~\ref{cor:lowerbound} immediately implies that every PMD code of length $n$, rate $R$, and error parameter $\epsilon$
must satisfy \[R \leq 1 - \frac{2 \log_q(1/\epsilon)}{n} + \frac{1 -\epsilon^2}{n\epsilon^2 q^{2n}\ln q},\]
which implies that, for any constant $\epsilon \in (0,1)$, 
\[ R \leq 1 - \frac{2}{n}\log_q\left(\frac{1}{\epsilon}\right) + o(1).\]

Bergamaschi~\cite{Ber24} presented a construction of an $(n+\ell, n - \ell, \epsilon)_q$-PMD code
based on the purity testing codes by~\cite{BCGST02}
for every prime $q$ and sufficiently large $n, \ell \in \N$,
where $\epsilon \leq \sqrt{(2n+1)q^{-\ell}}$.
The redundancy parameter $\lambda$ is equal to $2\ell$.
Plugging $\epsilon = \sqrt{(2n+1)q^{-\ell}}$ into the bound in
Corollary~\ref{cor:lowerbound}, we have
\[ \lambda \geq 2\log_q\sqrt{\frac{q^\ell}{2n+1}} - \frac{1 - (2n+1)q^{-\ell}}{(2n+1)q^{-\ell}q^{2n}\ln q} = \ell - \log_q(2n+1) - O(q^{\ell - 2n}). \]
Hence, there is a gap of $\ell + O(\log_q n)$ between the construction of~\cite{Ber24} and our lower bound.
Closing this gap appears to require techniques beyond the unitary 1-design
arguments used in our proof or improvements in purity-testing-based constructions.
For AMD codes, the classical counterparts of PMD codes,
the redundancy (or tag length) lower bound $2\log(1/\epsilon) - O(1)$ is known to be tight~\cite{CDFPW08,CFP13}.
The similarity between this classical bound and Corollary~\ref{cor:lowerbound}
suggests that optimal PMD code constructions may also exist.
Determining whether our lower bound is tight remains open.

\section*{Acknowledgments}
 This work was supported in part by JSPS KAKENHI Grant Numbers 
23H00468 and 24H00071.

\bibliographystyle{alpha}
\bibliography{mybib}

\end{document}